\documentclass[reqno]{amsart}
\usepackage{times,amssymb,amsmath,amsfonts,amsthm,bbm,mathrsfs,color,enumitem,bm}
\usepackage[normalem]{ulem}
\usepackage{hyperref}
\usepackage[mathscr]{eucal}
\usepackage[top=1.6in,bottom=1.5in,left=1.6in,right=1.6in]{geometry}
\allowdisplaybreaks
\newcommand\redout{\bgroup\markoverwith{\textcolor{red}{\rule[0.5ex]{2pt}{0.8pt}}}\ULon}

\newcommand{\new}[1]{{\color{black}#1}}

\newtheorem{theorem}{Theorem}
\newtheorem{lemma}[theorem]{Lemma}

\newtheorem{proposition}[theorem]{Proposition}

\newtheorem{corollary}[theorem]{Corollary}
\newtheorem{definition}{Definition}

\newtheorem{example}{Example}[section]
\newtheorem{remark}{Remark}[section]

\usepackage{tikz}
\usetikzlibrary{calc}

\tikzset{font=\large}

\newcommand\nc\newcommand
\nc\bfa{{\boldsymbol a}}\nc\bfA{{\boldsymbol A}}\nc\cA{{\mathscr A}}
\nc\bfb{{\boldsymbol b}}\nc\bfB{{\boldsymbol B}}\nc\cB{{\mathscr B}}
\nc\bfc{{\boldsymbol c}}\nc\bfC{{\boldsymbol C}}\nc\cC{{\mathscr C}}
\nc\bfd{{\boldsymbol d}}\nc\bfD{{\boldsymbol D}}\nc\cD{{\mathscr D}}
\nc\bfe{{\boldsymbol e}}\nc\bfE{{\boldsymbol E}}\nc\cE{{\mathscr E}}
\nc\bff{{\boldsymbol f}}\nc\bfF{{\boldsymbol F}}\nc\cF{{\mathscr F}}
\nc\bfg{{\boldsymbol g}}\nc\bfG{{\boldsymbol G}}\nc\cG{{\mathscr G}}
\nc\bfh{{\boldsymbol h}}\nc\bfH{{\boldsymbol H}}\nc\cH{{\mathscr H}}
\nc\bfi{{\boldsymbol i}}\nc\bfI{{\boldsymbol I}}\nc\cI{{\mathcal I}}
\nc\bfj{{\boldsymbol j}}\nc\bfJ{{\boldsymbol J}}\nc\cJ{{\mathscr J}}
\nc\bfk{{\boldsymbol k}}\nc\bfK{{\boldsymbol K}}\nc\cK{{\mathscr K}}
\nc\bfl{{\boldsymbol l}}\nc\bfL{{\boldsymbol L}}\nc\cL{{\mathscr L}}
\nc\bfm{{\boldsymbol m}}\nc\bfM{{\boldsymbol M}}\nc{\cM}{{\mathscr M}}
\nc\bfn{{\boldsymbol n}}\nc\bfN{{\boldsymbol N}}\nc\cN{{\mathscr N}}
\nc\bfo{{\boldsymbol o}}\nc\bfO{{\boldsymbol O}}\nc\cO{{\mathscr O}}
\nc\bfp{{\boldsymbol p}}\nc\bfP{{\boldsymbol P}}\nc\cP{{\mathscr P}}
\nc\bfq{{\boldsymbol q}}\nc\bfQ{{\boldsymbol Q}}\nc\cQ{{\mathscr Q}}
\nc\bfr{{\boldsymbol r}}\nc\bfR{{\boldsymbol R}}\nc\cR{{\mathscr R}}
\nc\bfs{{\boldsymbol s}}\nc\bfS{{\boldsymbol S}}\nc\cS{{\mathscr S}}
\nc\bft{{\boldsymbol t}}\nc\bfT{{\boldsymbol T}}\nc\cT{{\mathscr T}}
\nc\bfu{{\boldsymbol u}}\nc\bfU{{\boldsymbol U}}\nc\cU{{\mathscr U}}
\nc\bfv{{\boldsymbol v}}\nc\bfV{{\boldsymbol V}}\nc\cV{{\mathscr V}}
\nc\bfw{{\boldsymbol w}}\nc\bfW{{\boldsymbol W}}\nc\cW{{\mathscr W}}
\nc\bfx{{\boldsymbol x}}\nc\bfX{{\boldsymbol X}}\nc\cX{{\mathscr X}}
\nc\bfy{{\boldsymbol y}}\nc\bfY{{\boldsymbol Y}}\nc\cY{{\mathscr Y}}
\nc\bfz{{\boldsymbol z}}\nc\bfZ{{\boldsymbol Z}}\nc\cZ{{\mathscr Z}}
\nc{\bb}{{\mathbbm{1}}}
\nc\reals{{\mathbb R}}
\nc{\ff}{{\mathbb F}}
\nc{\PP}{{\mathbb P}}
\nc{\integers}{{\mathbb Z}}

\DeclareMathOperator{\rank}{rk}
\DeclareMathOperator{\supp}{supp}
\nc{\Cay}{{\sf Cay}}

\newcommand\remove[1]{}

\newcommand{\F}{\EuScript F}
\newcommand{\f}{\mathbb F}

\newcommand{\one}{\mathbf 1}
\newcommand\isomto{\stackrel{\sim}{\smash{\longrightarrow}\rule{0pt}{0.5ex}}}
\DeclareMathOperator{\Ima}{Im}

\DeclareSymbolFont{bbold}{U}{bbold}{m}{n}
\DeclareSymbolFontAlphabet{\mathbbold}{bbold}

\begin{document}
\title{High-rate storage codes on triangle-free graphs}
\author{Alexander Barg}\address{Department of ECE and Institute for Systems Research, University of Maryland, College Park, MD 20742, USA}\email{abarg@umd.edu}
\author{Gilles Z{\'e}mor}\address{Mathematics Institute, UMR 5251, University of Bordeaux, France}\email{Gilles.Zemor@math.u-bordeaux.fr}

\begin{abstract} Consider an assignment of bits to the vertices of a connected graph $G(V,E)$ with the property that the value of each 
vertex is a function of the values of its neighbors. A collection of such assignments is called a {\em storage code} of length $|V|$ on $G$. 
The storage code problem can be equivalently formulated as maximizing the probability of success in a {\em guessing game} on graphs, or constructing {\em index codes} of small rate.

If $G$ contains many cliques, it is easy to construct codes of rate close to 1, so a natural problem is to construct high-rate codes on 
triangle-free graphs, where constructing codes of rate $>1/2$ is a nontrivial task, with few known results. In this work we construct infinite families of linear storage codes with high rate relying on coset graphs of binary linear codes. We also derive necessary conditions for such codes to have high rate, and even rate potentially close to one. 

We also address correction of multiple erasures in the codeword, deriving recovery guarantees based on expansion properties of the graph.

Finally, we point out connections between linear storage codes and quantum CSS codes, a link to bootstrap percolation and contagion spread in
graphs, and formulate a number of open problems.
\end{abstract}
\maketitle

\section{Introduction}
In this paper we consider a class of codes on graphs known as {storage codes}. Given an undirected graph $G(V,E)$ with $N$ vertices, denote by $\cN(v)=\{u: (v,u)\in E\}$ the set of neighbors of the vertex $v$ in $G$. Consider a set of vectors $\cC=\{x: x\in Q^N\}$ over a finite alphabet $Q$, where the coordinates are indexed by the vertices in $V$. \new{We assume some fixed order of the vertices of $G$ throughout the paper.} The set $\cC$ is said to form a {\em storage code} if for every $v\in V$ and $x,y\in \cC,$ if $x_u=y_u$ for all $u\in \cN(v)$ then also $x_v=y_v.$ In other words, the codewords are written on the vertices of $G$, and for any codeword the value of the vertex $v$ can be uniquely determined by the values of its neighbors. Thinking of storing the coordinates of the codeword at different nodes of a distributed storage system, this definition implies that an erased coordinate (vertex) can be recovered in a local way from its immediate neighborhood. This definition can be also phrased in terms of recovery functions $f_v$ that map the subcodewords $(x_u)$, $u\in\cN(v)$, on the value $x_v$ of $v.$ Note that $f_v$ generally depends on $v$, and the functions for different vertices may be different.

The concept of storage codes on graphs was introduced around 2014 by Mazumdar \cite{Mazumdar2014,Mazumdar2015} and Shanmugam and Dimakis 
\cite{Shanmugam2014}, motivated by earlier works on index coding by Alon et al. \cite{AloLubStaWeiHas2008} and codes with locality for distributed storage introduced by Gopalan et al. \cite{gopalan2011locality}. The defining property of codes with locality is the ability to recover a missing coordinate of the codeword 
based on the values of the symbols in a small subset of other coordinates (the repair group), with the goal of minimizing the size of
these subsets and reducing internodal communication for completing the recovery task. Storage codes on graphs additionally assume that the links between the nodes are established based on physical proximity and the associated energy constraints, limitations of the system 
architecture, or other features with the same effect. This constraint, modeled by the graph $G$, defines the neighborhood of each vertex
and guides the construction of the code used to store information on the vertices.
Similarly, the problem of index coding addresses constructions of broadcast functions that distribute information to the vertices
of the graph whereby each vertex has access to the ``side information'' stored on the vertices in its neighborhood in the graph.

Essentially the same problem, in a different guise, appears in a line of works devoted to guessing games on graphs, e.g., \cite{CM2011,Cameron2016}, which has developed independently of both the storage codes and index coding problems. That these groups of problems are largely equivalent was realized in a number of papers,
and the historical development is detailed in \cite{AK2015} from the index codes' perspective. We present a brief summary of the relations between these three problems in Section~\ref{sec:3P} below.

\begin{example}\label{Example:5} \new{\rm To give an example, consider the graph in Figure~\ref{fig:graph}. We can construct a storage code by assigning to its vertices $1,2,\ldots,5$, 
binary vectors $(x_1,x_2,\dots,x_5)$ that satisfy $x_1=x_2$ and $x_3+x_4+x_5=0$. Clearly, every vertex
can recover its value from its neighborhood, for instance, the recovery functions of the vertices $1$ and $3$ are $f_1(\{x_2,x_3\})=x_2$ and $f_3(\{x_1,x_4,x_5\})=x_4+x_5$, respectively. All the possible assignments give rise to a linear binary storage code of length $N=5$ and dimension 3, so the rate of the
code is $R(G)=3/5.$}

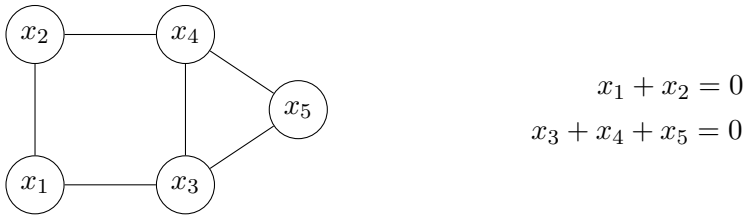
\begin{figure}[ht]
\begin{center}
\begin{tikzpicture}
[
vertex_style/.style={circle, draw,minimum size=0.01cm,scale=1.0},
vertex_style1/.style={circle, draw, scale=1.0},
edge_style1/.style={thick},edge_style2/.style={thick, black}]
 
\node[vertex_style]  (v1) at (-1,0.5)[draw] {$x_1$};
\node[vertex_style]  (v2) at (-1,2.5)[draw] {$x_2$};
\node[vertex_style]  (v3) at (1.0,0.5)[draw] {$x_3$};
\node[vertex_style] (v4) at (1.0,2.5)[draw] {$x_4$};
\node[vertex_style] (v5) at (2.5,1.5)[draw] {$x_5$};

\draw (v1) -- (v2) -- (v4)-- (v3) -- (v1);
\draw (v4) -- (v5) -- (v3);

\node at (7.1,1.8) {$\phantom{x_3+ }x_1+x_2=0$};
\node at (7,1.2) {$x_3+x_4+x_5=0$};

\end{tikzpicture}
\end{center}
\caption{Graph and storage code}
\label{fig:graph}
\end{figure}
\end{example}

The main question concerning storage codes is determining the largest possible cardinality of the code $\cC$ for a given graph $G$. Below we denote by $R_q(G)$ the maximum possible rate $\frac 1N \log_q|\cC|$ of a $q$-ary storage code for a given graph $G$. 
Determining $R_q(G)$ (or even $R(G):=\sup_q R_q(G)$ if we wish to optimize on the choice of the code alphabet) 
is a difficult problem related to the so-called minrank of the graph, and we refer to \cite{MazMcgVor2019} for the 
best known bounds as well as an overview of the earlier results. Below we limit ourselves to finite field alphabets $Q=\f_q$ and assume that the code $\cC$ forms a {\em linear subspace} of $\f_q^N.$ While nonlinear functions can sometimes attain higher rates \cite{AloLubStaWeiHas2008}, adding linearity enables one to rely on various known structures to construct families of storage codes. Generally, improved rates associated with nonlinear storage/index codes require specially designed dependence graphs, while linear maps support constructions for large graph families based on universal methods.

In the next section we overview the known constructions of storage codes. It turns out that in dense graphs, for 
instance, graphs with many cliques, attaining high rate is easy: \new{ in particular, if the graph $G$ on vertices $1,2,\ldots, n$ is itself a clique, then it affords a storage code of rate close to one defined by the single parity check equation $\sum_{i=1}^nx_i=0$, that is common to all the vertices. More generally, 
if we partition the vertex set of the graph into cliques, we can define a storage code with one single equation for every clique. In the example of Figure~\ref{fig:graph}, there are two equations, one for clique $\{1,2\}$ and one for clique $\{3,4,5\}$. Therefore graphs with many large cliques tend to admit storage codes of high rate.
But in graphs with no cliques of size larger than $2$, i.e., triangle-free graphs, 
the simple clique partition strategy achieves at best rate $1/2$. Other more refined constructions also fail to exceed rate $1/2$ when graphs are triangle-free
 (more on this below), and it was also proved that bipartite graphs, a large class of triangle-free graphs, do not admit storage codes of rates that exceed $1/2$. 
 This motivated
us to look for triangle-free graphs that do admit storage codes of rate larger than $1/2$, and this forms the main subject of this paper.
Constructing codes of high rate on such graphs represents a challenge, and early studies \cite{CM2011} have conjectured 
that with no triangles, $R=1/2$ is the largest attainable rate value. While this conjecture was refuted in \cite{Cameron2016}, only isolated examples 
of such codes (and their direct extensions to infinite families; see Remark~\ref{remark:copies} below) have been presented in the literature. }

The codes that we consider are constructed on Cayley graphs on $\f_2^r$ 
(coset graphs of binary linear codes). While coset graphs form a classic topic in coding theory, one 
motivation of this work is drawn from a recent construction of quantum CSS codes obtained from linear spaces defined by adjacency matrices
of coset graphs \cite{Couvreur2013}. 
Our main results are constructions of two infinite families of binary linear storage codes of rate above 1/2, one of which 
in fact has rate $R_2(G)$ approaching 3/4 for growing length $N$, exceeding the rate 5/8 of the unique previously known binary storage code with rate above 1/2 for triangle-free graphs. The exact dimension of these storage codes is computed.
We also formulate necessary conditions for the code 
dimension to be high compared to the number of vertices $N$. 

In addition, we present computer-assisted results that give sporadic examples of
storage codes with high rate on triangle-free coset graphs: in particular we uncover an example of a storage code on a triangle-free graph
whose rate exceeds 3/4 and surpasses that of all previously known storage codes over any alphabet.

Finally, we address a problem left open in \cite{Mazumdar2015}, namely that of correcting multiple erasures in the codeword: we derive local recovery guarantees based on expansion properties of the graph, and make a connection with {\em bootstrap percolation}.

The paper is organized as follows. Section~\ref{sec:background} is dedicated to background material on storage codes, with connections to index coding and guessing games. Section~\ref{sec:Coset graphs} introduces storage codes on coset graphs of linear codes and makes a connection with quantum coding,
transforming a known family of quantum codes into storage codes.
Section~\ref{sec:necessary} derives necessary conditions for storage codes on triangle-free coset graphs to have large rates.
Section~\ref{sec:3/4} is devoted to constructing a family of binary linear codes and proving that their coset graphs are triangle-free and admit storage codes of rate approaching 3/4. 
Section~\ref{sec:multiple} addresses the problem of correcting multiple erasures for storage codes on expander graphs. Finally,
Section~\ref{sec:discussion} concludes with results on some storage code rates for triangle-free coset graphs found with computer help, and with miscellaneous comments and open problems.

\section{Background}\label{sec:background}

\subsection{Storage codes, index codes, and guessing games}\label{sec:3P}
The problem of finding the largest storage code for the graph $G$ is closely related to two other recently introduced problems in computer science, {\em index coding} and {\em guessing games on graphs}. In the symmetric index coding problem, an information vector $x=(x_1,\dots,x_N)$ is to be distributed to $N$ users via a single broadcast with the goal of furnishing user $i$ with the symbol $x_i.$ The users possess ``side information'' about their
intended messages in the form of a subset of coordinates of the vector $x$. Specifically, let $N_i\subset\{1,\dots,N\}$ be a subset of indices, and suppose that user $i$ has access to symbols $x(N_i):=(x_j,j\in N_i)$. The index coding problem seeks to construct an encoding (broadcast) 
function $f:\f_q^N\to W,$ where $K:=\log_q|W|<N,$ and $N$ decoding functions $\phi_i$ such that $\phi_i(f(x),x(N_i))=x_i$ for all $i=1,\dots,N$.

The index coding problem is conveniently described in terms of a {\em side information graph} $G(V,E)$ with $N$ vertices and $(i,j)\in E$ if and
only if $j\in N_i.$ We assume that the relation $j\in N_i$ is symmetric in $i,j$, and thus the edges are 
undirected. 
The value $I_q(G,f):=K/N$ is called the rate of the $q$-ary symmetric index code $(f,\{\phi_i\})$, 
and the objective of the construction is to devise broadcast functions (index codes) of minimum rate for a 
given side information graph.  
The minimum possible rate of the index code for the side information graph $G$ is called the capacity $I_q(G)$ of index coding for $G$. It is often possible to reduce the rate value $I_q(G)$ by increasing $q$, resulting in
a universal characterization $I(G)$ of index coding for a given graph.

The problems of storage coding and index coding are to a large extent equivalent, namely, the following is true:
    \begin{enumerate}
\item[(A)] Any graph $G$ defines an index coding problem $(f,\{\phi_i\})$ wherein the sets $N_i=\cN(i)$ are the neighborhoods of vertices 
$i\in V$. Furthermore, for any
value $s$ of the broadcast function, the preimage $f^{-1}(s)$ is a storage code for the graph $G$.
\item[(B)] If there is a partition $(\cC_j)_{j\in W}$ of $\f_q^N$ into a set of storage codes $\cC_s$ for $G$, then the map 
$\f_q^N \to W$ that associates a vector $x$ to the index $s$ such that $x\in \cC_s$, is a valid broadcast 
function for the index coding problem with side information graph $G.$
\item[(C)] Moreover, if the broadcast function $f$ for an index coding problem defined by a graph $G$ is linear, then its kernel is a storage code. Conversely, if $\cC$ is an $\f_q$-linear $[N,K]$ storage code, then any syndrome function $f:\cC\to 
\f^{N-K}$ is a broadcast function for the corresponding index coding problem. This link is a starting point of the equivalence proof of the two problems in \cite{Mazumdar2014}.
\end{enumerate}
As shown in \cite{AloLubStaWeiHas2008,Mazumdar2015}   
   $$
   1-I_q(G)\le R_q(G)\le 1-I_q(G)+N^{-1}\log_q(N\ln q),
   $$
and thus $\sup_q(R_q(G))+\inf_q(I_q(G))=1$.

Another equivalent formulation of the storage coding problem arises from the study of guessing games 
on graphs initiated in \cite{Riis2007}. In one version of the game, the vertices in $V$ are assigned 
elements of a finite set $Q$ (colors), 
and each vertex $v$ attempts to guess its color based on the colors of its neighbors in $\cN(v).$ The
game is won if all the vertices correctly guess their colors. They may agree on the guessing 
strategy in advance, but once they are assigned colors, all communication is forbidden. 
The assignment $x\in Q^N$ is assumed uniformly random, and the participants (vertices) attempt 
to maximize their probability of success $P_s.$ Any storage code $\cC(G)$ for the graph $G$
defines a guessing strategy, namely every participant will assume that $x\in\cC$ and reconstruct their $x_v$ from the values of their neighbors:
we thus have $P_s(\cC,G)=\frac{|\cC(G)|}{q^N}=\exp(\log|\cC|-N).$ Conversely, any guessing strategy defines a storage code, namely the set of values $x$ for which the strategy succeeds, 
and thus maximizing $P_s$ is equivalent to constructing large-size storage codes. The authors of \cite{Cameron2016} 
define the {\em guessing number} of $G$ as 
     $
     {\sf gn}_q(G)=N+\log_q\max P_s(\cC,G);
     $
thus in our notation ${\sf gn}_q(G)=N R_q(G).$ In the context of guessing games it has been shown that the storage
capacity $R_q(G)$ is almost monotone on $q$. Namely, the following is true.
\begin{theorem}[\cite{Gadouleau2011,CM2011,Cameron2016}] For any graph $G$, alphabet size $q,$ and $\epsilon>0$
there exists $q_0(G,q,\epsilon)$ such that for all $q'>q_0$
   $$
   R_{q'}(G)\ge R_q(G)-\epsilon.
   $$
Consequently, the limit $\lim_{q\to\infty} R_q(G)$ exists.
\end{theorem}

The equivalence between index coding, storage coding, and guessing games is further discussed in \cite{AK2015}.

\subsection{Constructions of storage codes}\label{sec:constructions}
To motivate the problem that we will study in the next section, we briefly overview the known constructions; see \cite{CM2011,Cameron2016,MazMcgVor2019} and also \cite[Ch.6]{AK2018} for the details. 

\subsubsection {Matching construction} A matching $M\subset E$ in a graph $G(V,E)$ \new{is a subset of edges such that
every vertex $v\in V$ is incident to at most on edge from $M$, and a matching is {\em perfect} if every $v\in V$
is incident to exactly one edge from it.} A matching $M$ defines a linear storage code $\cC\colon\f_q^{|M|}\to \f_q^N$ by
associating to $(x_e)_{e\in M}$ the vector $c\in \f_q^N$ such that for every vertex $v$ incident to an edge $e\in M$ we assign
$c_v=x_e$ and for every remaining vertex $v$ we put $c_v=0.$ In particular, if $G$ is $d$-regular and bipartite, it contains
a perfect matching, giving rise to a storage code $\cC(G)$ of rate $1/2$ independently of $q.$ If $G$ is not bipartite, we can take a double cover\footnote{\new{A bipartite double cover of $G$ is a graph with adjacency matrix
$\left(\!\!\!\begin{array}{c@{\hspace*{1pt}}c}0&A\\A&0\end{array}\!\!\!\right)$, where $A$ is the adjacency matrix of $G$.}} $\bar G$ of $G$.
Now, $\bar G$ is a bipartite regular graph, and we can apply the matching encoding and obtain a code $\bar\cC\subset \f_q^{2N}.$ This code can be mapped on
a code $\cC$ over $\f_{q^2}^N$ obtained by grouping together pairs of symbols indexed by vertices $v',v''\in V(\bar G)$
that correspond to the same vertex $v\in V(G).$ This construction gives a way of obtaining a code  of rate $1/2$ on any regular graph. Clearly, 
the value of any vertex in the codeword is recoverable from its neighbors, implying that $\cC$ is a storage code for $G.$

\subsubsection{Edge-vertex construction} An alternative way of obtaining a code of rate $1/2$ on a $d$-regular graph $G(V,E)$ is the following.
Consider the space $(\f_q)^{|E|}$ of vectors indexed by the edges of $G$. Next, map this space on $(\f_q^d)^N\cong \f_{q^d}^N$ by assigning to
each vertex $v$ a $d$-vector formed of the symbols written on the edges incident to it. In other words, $\cC(G)=\{(c_v, v\in V)\}$ is obtained
as the image of the mapping
      \begin{equation}\label{eq:EV}
      \begin{aligned}
      \f_q^{|E|}&\to \f_{q^d}^N\\
       (x_e)_{e\in E} & \mapsto (c_v)_{ v\in V}, \text{ where }c_v=(x_e)_{e\in \partial(v)},
       \end{aligned}
       \end{equation}
where $\partial(v)$ is the edge neighborhood of $v$ in $G.$ Since $|(\f_q)^{|E|}|=q^{\frac{dN}2},$ the rate of $\cC$ is indeed 1/2. \new{In Figure~\ref{fig:edgevertex} we show this construction for the cycle $C_5,$ where we first place symbols of the $q$-ary alphabet on the edges, and then assign to each vertex the symbols on the edges that are incident to it. Now each vertex can recover its pair of symbols from its neighbors: for instance, $v_1$ takes the second symbol from $v_2$ and the first one from $v_5$ (note again that we fix the order of the vertices). This gives a code of size $|\cC|=q^{|E|}$ with $|E|=5,$ and thus the rate is $R(\cC)=\frac1N\log_{q^2} |C|=1/2.$  }

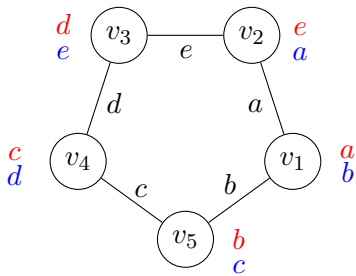
\begin{figure}
\begin{center}\begin{tikzpicture}[rotate=-18]
\node (v1) at ( 0:1.5) [draw,shape=circle] {{$v_1$}};
\node (v2) at ( 72:1.5) [draw,shape=circle] { $v_2$};
\node (v3) at (2*72:1.5) [draw,shape=circle] {$v_3$};
\node (v4) at (3*72:1.5) [draw,shape=circle] {$v_4$};
\node (v5) at (4*72:1.5) [draw,shape=circle] {$v_5$};
\draw (v1) -- (v2) -- (v3)-- (v4)-- (v5)--(v1);
\node at (0.8,0.55) {$a$};
\node at (0.8,-0.55) {$b$};
\node at (-0.3,-1) {$c$};
\node at (-1,0.0) {$d$};
\node at (-0.3,0.95) {$e$};
\node at (1.1,1.55) {$\begin{matrix} \textcolor{red}{e}\\[-.05in]\textcolor{blue}{a} \end{matrix}$};
\node at (2.2,0.2) {$\begin{matrix} \textcolor{red}{a}\\[-.05in]\textcolor{blue}{b} \end{matrix}$};
\node at (1.2,-1.4) {$\begin{matrix} \textcolor{red}{b}\\[-.05in]\textcolor{blue}{c} \end{matrix}$};
\node at (-2.0,-1.2) {$\begin{matrix} \textcolor{red}{c}\\[-.05in]\textcolor{blue}{d} \end{matrix}$};
\node at (-1.9,0.6) {$\begin{matrix} \textcolor{red}{d}\\[-.05in]\textcolor{blue}{e} \end{matrix}$};
\end{tikzpicture}
\end{center}
\vspace*{-.15in}\caption{Edge-vertex construction of a storage code}
\label{fig:edgevertex}
\vspace*{.15in}
\end{figure}

\subsubsection{Clique-vertex construction} The edge-vertex construction affords a straightforward generalization if every vertex $v\in V$ is incident to the same number (say $m$) of $k$-cliques, where $k>2.$ Let $\cK$ be the set of $k$-cliques in $G$ and let us map $(\f_q^{k-1})^{|\cK|}$ to $(\f_q^m)^n$ by placing a $q$-ary code of length $k$ and dimension $k-1$ on every clique and distributing the symbols of the $k$-codeword
to the vertices that form the clique. The rate of the code $\cC(G)$ obtained as a result equals $(k-1)/k;$ for instance, a triangular lattice gives rise to a code of rate $2/3,$ etc. 
1
\subsubsection{Fractional clique cover} A {fractional clique cover} of a graph $G$ is a collection $\cK$ of
cliques $\kappa$ together with a weight $w_\kappa\in [0,1]$, such that for every vertex
$v$, we have 
   \begin{equation}\label{eq:fc}
\sum_{\kappa\in\cK: v\in \kappa} w_{\kappa}\geq 1
   \end{equation}
To construct a storage code for $G$, we apply the clique-vertex construction to the subset of cliques that make up the fractional clique cover. The vertices may not all be incident to the same number of cliques, in which case the alphabet size may depend on the vertex (a mixed-alphabet code), or we have to use the alphabet $\f_q^m$ where $m$ is the largest vertex-degree in the vertex to clique incidence graph. 

\subsection{Bounds for storage codes}
It is known that $R_q(G)$ satisfies the constraints
  \begin{equation}\label{eq:VC}
  M(G)\leq N R_q(G)\leq N-\alpha(G)
  \end{equation}
where $M(G)$ is the size of the largest matching in $G$ and $\alpha(G)$ is the \new{size of the largest independent set} in  $G$, \new{i.e., vertices with no edge among them}. This result was proved in \cite{Mazumdar2015} for storage codes and in an earlier independent work \cite{CM2011} in the language of guessing games. Since the 
complement of a maximum independent set forms a vertex cover in $G$ \new{(a set of vertices that touches every edge in
$E(G)$)}, the upper bound in \eqref{eq:VC} is often stated in terms of the vertex cover number. 
\new{A consequence of \eqref{eq:VC} is that if the graph $G$ is not only triangle-free but bipartite, we must have $R\leq 1/2$ since a bipartite graph has an independent set of size at least $N/2$.}

To state another bound, recall that a fractional clique cover is called {\em regular} if its weighting $w:K(G)\to[0,1]$ satisfies
 \eqref{eq:fc} with equality.
As shown in \cite{CM2011} (for the guessing number) and in \cite{BKL2013} (for the broadcast rate), 
  $$
  R(G)\geq 1-\frac{1}{N}\kappa_f(G),
  $$
where $\kappa_f(G):=\min\sum_{\kappa}w_\kappa$ is the minimum weight of a regular fractional cover. The authors of \cite{CM2011} observed that this bound holds with equality for perfect graphs and cycles or their complements.

A linear program for bounding $R_q(G)$ was introduced by the authors of \cite{MazMcgVor2019},
adapting a similar bound for $I_q(G)$ \cite{BKL2013} to the case of storage codes. As a result, \cite{MazMcgVor2019} showed an upper bound $R_q(G)\le 1/2$ for several classes of graphs. 

Below we focus on the question of identifying graphs that support storage codes of the largest rate. Without any constraints, the complete graph $K_N$
affords a linear storage code of rate $\frac{N-1}N$ (defined by the single equation $\sum_{v\in V}x_v=0$). The problem becomes less trivial if we limit ourselves
to sparse graphs $G$, for instance, $d$-regular graphs on $N$ vertices with varying $N$.
\begin{proposition}\label{prop:B} Let $G$ be a connected $d$-regular graph on $N$ vertices, then
    $
    R_q(G)\le \frac{N+1}N-\frac 1{d+1}.
    $
There exist $d$-regular graphs with $R_q(G)\ge 1-\frac 1{d-2}.$ 
\end{proposition}
\begin{proof} \new{Recall that the Cartesian product of graphs $G$ and $H$ is a graph with the set of vertices
$V(G)\times V(H)$ in which vertices $(u,u')$ and $(v,v')$ are adjacent if and only if either
$u=v$ and $(u',v')\in E(H)$ or $(u,v)\in E(G)$ and $v=v'.$} \new{Take $s\ge 3$ and} let
$G$ be the Cartesian product of $K_{d-2}$ and a cycle of length $s$. \new{According to the above definition, the graph obtained as a result is formed of $s$ copies of the 
graph $G$ in which copies of the same vertex are connected by the edges of the cycle.} Now put a code of length $d-2$ on each clique independently, obtaining the claimed lower bound. On the
other hand, any $d$-regular graph \new{on $N$ vertices} contains an independent set of size $\lfloor\frac N{d+1}\rfloor$ 
\cite{Rosenfeld1964}, and together with \eqref{eq:VC} this implies the upper bound.
\end{proof}

\begin{remark}\label{remark:copies} {\rm This last proposition highlights the following fact. If we have a graph $G$ that admits a storage code $C$ of rate $R_q(G)$, we can construct many more graphs with the same storage rate, simply by taking $k\ge 1$ disjoint copies of the original graph $G$. The associated storage code will consist of the Cartesian product of $k$ copies of the original code $C$, whose codewords consist therefore of $k$ successive codewords of $C$. We can also add arbitrary edges to and between the copies of $G$ that the storage code will simply ignore. In what follows, we will be looking for the highest possible rate of a storage code for graphs with specific properties, and will conflate a code with its successive Cartesian powers, considering that we are dealing with the same code `up to repetitions'. }
\end{remark}

\section{Storage codes on triangle-free graphs and the CSS code connection} \label{sec:Coset graphs}

The graphs used to obtain storage code rates close to one are dense and contain a large number of cliques. It is therefore natural to consider the question of the largest attainable rate for graphs that contain
no cliques, i.e., triangle-free graphs. For such graphs none of the methods mentioned above yield a 
storage code rate in excess of $1/2$; in particular, as noted in \cite{BKL2013} (in the language of index codes),  $\kappa_f(G)\ge N/2$.
The same observation was also made in \cite{Cameron2016} in terms of the guessing number of triangle-free graphs. Both \cite{BKL2013} and \cite{Cameron2016} reported only one example of a storage code on triangle-free graphs for which $R_2(G)>\frac 12$ (up to repetitions, in the sense of Remark~\ref{remark:copies}).
This code is associated to the graph $\Gamma$ shown below in Fig.~\ref{fig:Clebsch} (the authors of \cite{Cameron2016} also gave several examples of triangle-free graphs with $R_q(G)>\frac 12$ for $q>2$). This graph can be defined in several ways and
is known as the {\em Clebsch graph} or a {\em folded cube}, see Figure~\ref{fig:Clebsch}. The storage code constructed on $\Gamma$
is a linear binary code $\cC$ of length $N=16$ whose parity-check matrix has the form
  \begin{equation}\label{eq:H}
  {\tilde A}=I_{N}+A(\Gamma),
  \end{equation}
where $A(\Gamma)$ is the adjacency matrix of $\Gamma.$ 
The identity matrix $I_N$ is added since $A$ includes only the neighborhood $\cN(v)$ but not the vertex $v$ itself. Upon checking that the $\f_2$-rank of ${\tilde A}$ equals $6,$ we conclude that the dimension of the code equals $N-\rank {\tilde A}=10,$ 
so $R_2(G)\ge\frac {5}{8}.$ This example refuted a conjecture in \cite{CM2011} which suggested that the fractional clique cover bound holds with equality for triangle-free graphs. It is also easy to check that $\Gamma$ contains independent sets of size $5$, and thus from \eqref{eq:VC} we conclude that $\frac58\le R_2(\Gamma)\le \frac{11}{16},$ where the upper bound is sharper than the result of Prop.~\ref{prop:B}. Using the terminology of guessing games, the guessing number of the graph $\Gamma$ satisfies $10\leq \sf{gn}_2(\Gamma)\leq 11,$ or rephrasing again, the side information graph $\Gamma$ affords a binary index code of rate $3/8.$

Note that the recovery functions of the vertices use full neighborhoods (the parity-check equations have weight 6), even though the general definition of the storage code does not include this constraint. We call binary codes from this subclass {\em full-parity} storage codes,
and it is only such codes that we consider in this and the next sections. 

\begin{figure}[ht]
\vspace*{.5in}
\begin{center}\scalebox{0.5}{
\begin{tikzpicture}
[
vertex_style/.style={circle, draw, fill,minimum size=0.01cm,scale=0.8},
vertex_style1/.style={circle, draw, scale=1.0},
edge_style1/.style={thick},edge_style2/.style={thick, black}]
 
\useasboundingbox (-5.05,-4.4) rectangle (6,6);
 
\begin{scope}[rotate=90]
 \foreach \bfx/\bfy in {0/1,72/2,144/3,216/4,288/5}{
 \node[vertex_style] (\bfy) at (canvas polar cs: radius=4.5cm,angle=\bfx){};
 }
 \foreach \bfx/\bfy in {0/6,72/7,144/8,216/9,288/10}{
 \node[vertex_style] (\bfy) at (canvas polar cs: radius=7.5cm,angle=\bfx){};
 }
 \foreach \bfx/\bfy in {36/11,108/12,180/13,252/14,324/15}{
 \node[vertex_style] (\bfy) at (canvas polar cs: radius=3cm,angle=\bfx){};
 }

\node[circle,draw,fill,double] (16) at (canvas polar cs: radius=0cm,angle=0){};

\path[-]  (12) edge [bend left=30,cyan,ultra thick] (6);\path[-]  (11) edge [bend right=30,cyan,ultra thick] (8);

\foreach \bfx/\bfy in{7/2,1/3,13/4,14/16,15/5,9/10}{
\path[-] (\bfx) edge [cyan,ultra thick] (\bfy);
}

\path[-] (7) edge [bend left=30,blue,ultra thick] (15);\path[-] (11) edge [bend left=30,blue,ultra thick] (10);
\foreach \bfx/\bfy in {2/5,6/1,12/3,16/13,14/4,8/9}
\path[-] (\bfx) edge [blue,ultra thick] (\bfy);

\path[-] (13) edge [bend right=30,olive,ultra thick] (10);\path[-] (8) edge [bend right=30,olive,ultra thick] (14);
\foreach \bfx/\bfy in {3/5,12/2,16/11,15/1,6/7,4/9}
\path[-] (\bfx) edge [olive,ultra thick] (\bfy);

\path[-] (13) edge [bend left=30,orange,ultra thick] (7);\path[-] (12) edge [bend right=30,orange,ultra thick] (9);
\foreach \bfx/\bfy in {2/4,14/5,16/15,11/1,6/10,3/8}
\path[-] (\bfx) edge [orange,ultra thick] (\bfy);

\path[-] (14) edge [bend right=30,ultra thick] (6);\path[-] (15) edge [bend left=30,ultra thick] (9);
\foreach \bfx/\bfy in {3/13,12/16,2/11,4/1,5/10,7/8}
\path[-] (\bfx) edge [ultra thick] (\bfy);

\end{scope}
\node[draw=none] at (4.5,-6.5) {0000};
\node[draw=none] at (8,2) {0001};
\node[draw=none] at (0,8) {1001};
\node[draw=none] at (0.8,4.8) {1101};
\node[draw=none] at (-8,2) {1011};
\node[draw=none] at (-5,-6.5) {0100};
\node[draw=none] at (-3.5,-3.5) {1100};
\node[draw=none] at (3.5,-3.5) {0010};
\node[draw=none] at (0,-3.5) {0011};
\node[draw=none] at (0,-.5) {0111};
\node[draw=none] at (-3.7,-1) {1000};
\node[draw=none] at (3.7,-1) {0110};
\node[draw=none] at (4.5,1.9) {1110};
\node[draw=none] at (2.3,2.8) {1111};
\node[draw=none] at (-2.3,2.8) {0101};
\node[draw=none] at (-4.5,1.9) {1010};

\end{tikzpicture}}
\end{center}
\vspace*{.5in}\caption{The Clebsch graph $\Gamma$ (a folded cube): A 5-regular triangle-free graph on $n=16$ vertices obtained by
identifying each pair of opposite vertices in the 5-dimensional hypercube.
It is also a coset graph of the binary code $\{00000,11111\}$.} \label{fig:Clebsch}
\vspace*{.1in}\end{figure}
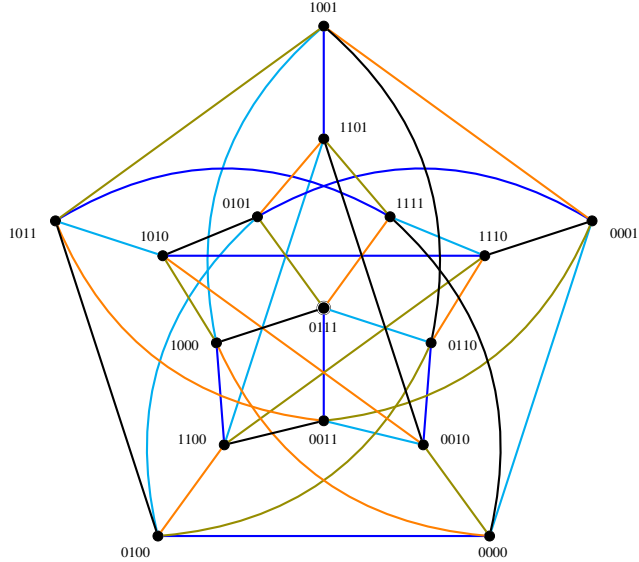

The Clebsch graph forms a rather special example: it is a unique strongly regular graph with the parameters $(16,5,0,2)$; see \cite{Godsil2001}, Theorem 10.6.4. Six other triangle-free strongly regular graphs are known \cite[p.~119]{Brouwer2012}, and all the other examples of (nonbinary) storage codes of rate greater than $\frac12$ in \cite{Cameron2016} are drawn from this list.

\subsection{Storage codes on Cayley graphs}\label{sec:Cayley}

Recall that the Cayley graph $\Cay(\cG,S)$ of the group $\cG$ for a given set of generators $S$ has $\cG$ as its set of vertices, and the vertices 
$g_1$ and $g_2$ are connected by an edge if there is an element $s\in S$ such that $g_1s=g_2.$ For the group
$\cG=\f_2^r$, which is the only group we will consider, any subset $S$ coincides with its inverse $S^{-1},$ so the graphs that we study are undirected. Since the generators are $r$-dimensional binary vectors, we also write the group action additively. Write the elements of $S$ as columns of an $r\times n$ matrix $H$ and consider the binary code $C$ defined by $H$ as the parity-check matrix. The graph
$\Cay(\cG,S)$ can be also viewed as a {\em coset graph} of the code $C$ constructed on the set of cosets in $\f_2^n/C$,
wherein two cosets are connected with an edge if and only if the Hamming distance between them (as subsets) is one.

Observe that the Clebsch graph $\Gamma$ can be viewed as a Cayley graph over the group $\f_2^4$. Namely,
consider the set $S$ whose elements form the matrix 
        \begin{equation}\label{eq:rc}
    H=\begin{bmatrix}1&0&0&0&1\\0&1&0&0&1\\0&0&1&0&1\\0&0&0&1&1\end{bmatrix}.
    \end{equation}
Indeed, as easily checked, the canonical generators $e_1,e_2,e_3,e_4$ connect a vertex $x\in \f_2^4$ to its neighbors at Hamming distance one, while their sum $e_1+\dots+e_4$ connects it to its opposite vertex. See Fig.~\ref{fig:Clebsch} for one possible vertex
labeling, where each color corresponds to the action of a specific generator. Viewing the matrix $H$ as a parity-check matrix of a $[5,1]$ binary repetition code, we see that $\Gamma$ is in fact the {\em coset graph} of this code. 

Motivated by this example, we now investigate other Cayley graphs over binary groups $\f_2^r$ and exhibit more exceptional graph families
that admit binary storage codes of rate greater than $1/2.$ These graphs have connections to both classical and quantum coding theory.

\subsubsection*{Notation} Our notation is as follows. We are given a binary code $C$ with parameters $[n,k,d]$ with a fixed parity-check matrix $H$
which we are free to choose. Next we construct the coset graph $G=\Cay(\f_2^r,S)$ where $r=n-k$ is the number of
rows and $S$ is the set of columns of $H$. Note that $G$ is a regular graph of degree $|S|=n.$ 
The adjacency matrix $A$ of this graph is symmetric, of order $N=2^r$, and its rows and columns are labeled by the vectors $x\in \f_2^r$: 
we would like row 
$x$ to specifiy the parity-check equation that recovers the value of the vertex $x$ from its neighbors. Note however that $A_{x,x}=0$, therefore, to involve
the value supported by $x$ in the parity-check equation, we consider, as in \eqref{eq:H}, the matrix
   \begin{equation}\label{eq:fH}
   {\tilde A}=I_N+A(G).
   \end{equation}
Our storage code is finally a linear space in $\f_2^N$ defined as ${\cC}=\ker({\tilde A})$\footnote{Notice that
we deal with two types of binary codes: the codes in $\f_2^n$ and codes in $\f_2^N,$ 
called small codes and big codes in \cite{Couvreur2013}.}. Note that the matrix ${\tilde A}$ is the adjacency
matrix of the graph $\Cay(\f_2^r,S)$ to which we have added self-loops at every vertex, or equivalently the Cayley graph
$\Cay(\f_2^r, S'),$ where $S'$ is obtained from $S$ by adding 0 to the set of generators. The main problem addressed below is analyzing
the dimension of ${\cC}$ both in general and for several specific constructions.
   
As our first observation, note that once the minimum distance of the code $C$ is at least 4, then so is the girth of the graph 
$\Cay(\f_2^r,S),$ i.e., the graph is triangle-free. We will therefore assume that all the small codes below have distance 4 or more.

The next lemma will help to further motivate our problem.
\begin{lemma}\label{lem:oddeven} Let $A$ be the adjacency matrix of the graph $\Cay(\f_2^r,S)$ where $S$ may or may not contain $0$. 
If $n=|S|$ is odd then $\rank A=N=2^r$ and if $n$ is even then $AA^\intercal=0$, implying in particular 
$\rank A\leq N/2.$ 
\end{lemma}
\begin{proof}
The rows and columns of $A$ are indexed by vectors of $\f_2^r$. Two distinct rows
$x$ and $y$ intersect in the set of positions $\cI=(x+S)\cap (y+S)$ which is of even size because
whenever $x+s_1=y+s_2$ is in $\cI$, so is $x+s_2=y+s_1$, implying that $\cI$ is
partitioned into pairs of the form $\{z,z+(x+y)\}$. Therefore the matrix
$AA^\intercal$ has only zeros outside the main diagonal. Now a diagonal element has
value $\langle x,x\rangle =|S| \bmod 2$ so when $n$ is odd we have
$AA^\intercal = I_N$ and when $n$ is even we have $AA^\intercal=0$, hence the
claims of the lemma. 
\end{proof}
As a consequence, we observe that for a set of non-zero generators of {\em odd size} $n$,  we have $\rank\tilde{A}\le N/2$, whence $\dim\ker(\tilde{A})\ge N/2$
so the rate of the storage code satisfies $R({\cC})\ge 1/2.$ It may happen that for some
Cayley graphs of odd degree (not counting the loops) we have $\rank {\tilde A}<N/2$ in which case we will obtain an
exceptional storage code of large rate. While this observation shows that this construction has a potential for uncovering large-size storage codes, it does not necessarily make finding such codes a straightforward task: in fact, as we have mentioned, for some years it was believed that they did not exist at all.

\subsection{Coset graphs of binary codes: the quantum coding connection}
\new{The construction of quantum codes of Calderbank, Shor, and Steane \cite{Calderbank1996,Steane96} relies 
on a pair of classical codes $\cC_0,\cC_1$ such that $\cC_0^\bot\subset \cC_1$, and it gives rise to a quantum code
of dimension equal to $\dim(\cC_1/\cC_0^\bot)$ and distance $\min(w(\cC_1\backslash \cC_0^\bot), w(\cC_0\backslash \cC_1^\bot))$, where $w(\cdot)$ is the minimum Hamming weight of the argument set. Of interest to us is the particular case of $\cC_0=\cC_1$. In other words, we start with a binary linear code $\cC\subset\f_2^N$ generated by a 
matrix $A$ such that $AA^\intercal=0$, which ensures that $\cC^\bot\subset \cC$. The code $\cC$ defines a quantum code} through the above construction, and the associated quantum code has parameters $[[N,N-2\rank{A},D]]$, i.e. length $N$, dimension $N-2\rank A$ and minimum distance $D$ equal to the smallest Hamming weight of a vector in $\cC^\perp\setminus \cC$. 

Therefore, for Cayley graphs over $\f_2^r$, the storage codes we are interested in also define quantum codes. From the point of view of storage, we do not have much use for the quantum code's minimum distance, but we are very much interested in its dimension: in particular we will obtain a storage code of rate greater than $1/2$ if and only if the associated quantum code has non-zero dimension. Quantum codes arising from Cayley graphs over $\f_2^r$ were studied in \cite{Couvreur2013} and we will make use of some of results obtained in that work.

We now focus on the case of coset graphs of repetition codes of odd length $n=r+1$, since they form a natural generalization of the Clebsch graph which is the coset graph of the $[5,1]$ repetition code. Alternatively, they can be seen as Cayley graphs $\Cay(\f_2^r,S)$, $r$ even, with generating set $S=\{e_1,e_2,\ldots ,e_r, e_1+e_2+\cdots +e_r\}$, 
where $e_1,\ldots ,e_r$ denote the canonical generators. Now it turns out that in the quantum coding context,
coset graphs $G_m$ of repetition codes of {\em even} length $m$ were studied in \cite{Couvreur2013}, because their adjacency matrices $A(G_m)$ are self-orthogonal and therefore directly yield quantum codes. These quantum codes were proved in \cite{Couvreur2013} to have parameters $
[[2^{m-1},2^{m/2},2^{m/2-1}]]$. We shall use this result to prove
\begin{proposition}\label{prop:Clebsch}
The storage codes associated to coset graphs of repetition codes of odd length $n\ge 5$ have length $N=2^{n-1}$ and dimension 
$K=2^{n-2}+2^{(n-3)/2}$. 
\end{proposition}
\begin{proof}
We start with the graphs $G_m$ defined in the previous paragraph for even $m$. They are Cayley graphs with vertex set 
$\f_2^{m-1}$, generator set $S_m=\{e_1,\ldots ,e_{m-1}, e_1+\ldots +e_{m-1}\},$ and we note that all these generators are of odd weight
since $m$ is even.
Therefore $G_m$ is bipartite, and each of its edges connects an even-weight vertex to an odd-weight vertex. 
The adjacency matrix of $G_m$ can therefore be written as
     \begin{equation*}
    {A}(G_m) = \begin{bmatrix}
       0 & B_m^\intercal\\
       B_m & 0
           \end{bmatrix}
     \end{equation*}
where the rows of $B_m$ are indexed by even-weight vectors of length $m-1,$ its columns by odd-weight vectors, and $B_m$ has a $1$ in the coordinate indexed by row $x$ and column $y$ if and only if the vector $x+y\in S_m.$ From the result of \cite{Couvreur2013} quoted before the proposition we find that $\rank A(G_m)=\frac 12(2^{m-1}-2^{m/2})=2^{m-2}-2^{\frac m2-1}.$

We will now make a transition to odd length $n=m-1$ by transforming the odd-weight vector indices of the columns of $B_m$ into even-weight vectors through the correspondence $x\mapsto \mathbf{1}+x$, where $\mathbf{1}$ denotes the all-one vector of length $n=m-1$. 
After permuting columns, the matrix $B_m$ can now be seen as the adjacency matrix of the Cayley graph defined over the binary group of even-weight vectors of length $m-1$, with generators $\mathbf{1}+e_1,\mathbf{1}+e_2,\ldots ,\mathbf{1}+e_{n}, 0$. Since $n$ is odd we have
\[
\sum_{i=1}^{n}(\mathbf{1}+e_i)=0
\]
and this last Cayley graph is therefore isomorphic to the Cayley graph over $\f_2^{n-1}$ with generators 
   $$
0,e_1,\ldots ,e_{n-1},e_1+\cdots +e_{n-1}.
   $$   
The matrix $B_m$ is therefore the parity-check matrix of the storage code ${\cC}$ associated to the coset graph of the repetition code of {\em odd} length $n=m-1$ (note that this matrix has a unit diagonal). The length of the code ${\cC}$ is $N=2^{n-1},$ and the dimension 
   \begin{align*}
   K&=\dim\ker(B_m)=N-\frac12\rank {A}(G_m)=2^{n-1}-\frac12(2^{n-1}-2^{\frac{n+1}2-1})\\
       &=2^{n-2}+2^{(n-3)/2}
    \end{align*}
as was to be proved.    
\end{proof}

This last proposition therefore yields an infinite family of graphs for which the associated storage codes have rate exceeding $1/2$.
In particular, for $n=5$ we recover the value of the rate $\frac58$. As $n$ increases, the rate
   $$
   R(G_m)=\frac 12+\frac1{2^{(n+1)/2}}
   $$
decreases to 1/2, so the highest rate within this family is achieved for $n=5$, i.e. for the Clebsch graph.   

\section{Necessary conditions for high rate of codes on coset graphs}\label{sec:necessary}
We maintain the notation and conventions of Sec.~\ref{sec:Cayley}. Below we identify binary vectors $z\in \f_2^{2^r}$ with functions
    \[
f: \f_2^r\rightarrow \f_2,
    \]
namely we put $f(x)=z_x,$ where $x\in \f_2^r$ is the (vector) index of the coordinates of $z.$
Therefore the adjacency matrix $A$ can be viewed as an operator on the function space $\F(\f_2^r):=\f_2^{\f_2^r}$ acting by
  \begin{eqnarray*}
A\colon \F(\f_2^r) & \rightarrow &  \F(\f_2^r) \\
            f   & \mapsto     & Af
\end{eqnarray*}
with 
\begin{equation}
\label{eq:adjacency}
Af(x) = \sum_{s\in S}f(x+s).
\end{equation} 
Keeping the graph $G=\Cay(\f_2^r,S)$ in mind, we also call $A$ the {\em adjacency operator} of $G$.

The starting observation in this part is given by the following proposition, due to Lowzow~\cite{lowzow}, as is point (a) of Theorem~\ref{cor:R>1/2} below, that was formulated in a quantum coding context.

\begin{proposition}
\label{prop:odd}
Let $V$ be a vector subspace of $\f_2^r$. If $|S\cap V|$ is odd, then
$\rank{A}\geq 2^{\dim V}$.
\end{proposition}

\begin{proof}
Let us look at the effect of $A$ on the restricted set of functions $F_V:=\{f: \supp(f)\subset V\}.$ 
We note that for every $x\in V$ we have
\[
Af(x) = \sum_{s\in S\cap V}f(x+s)
\]
since for $x\in V$ and $s\not\in V$ we have $x+s\not\in V$ and $f(x+s)=0$ by our
hypothesis on $f$. Therefore, when thus restricted to functions $V\to \f_2$, 
$A$ acts on $F_V$ as the adjacency operator $A_{S\cap V}$ of the Cayley graph
defined by the set of generators $S\cap V$. Therefore the image of $A$ has
dimension at least equal to the dimension of the image of $A_{S\cap V}$. Now,
Lemma~\ref{lem:oddeven} says that when $|S\cap V|$ is odd the image of $A_{S\cap
V}$ has dimension $2^{\dim V}$, since $A_{S\cap V}$ is full-rank on the space of
functions $V\to\f_2$.
\end{proof}

\begin{definition} The {\em Schur product} of two codes $A,B\subset \f_2^n$ is a
binary linear code $C=A\ast B\subset \f_2^n$ generated by all coordinatewise products $a\ast b=(a_1b_1,\ldots,a_nb_n), a=(a_1,\ldots ,a_n)\in A,b=(b_1,\ldots,b_n)\in B.$
\end{definition}

\begin{theorem}\label{cor:R>1/2}
Let $G=\Cay(\f_2^r,S)$ be a Cayley graph over $\f_2^r$ with set of generators $S$ containing
the $0$ element. Let $H$ be the $r\times n$ matrix whose columns are made up of
the non-zero elements of $S$, so that $n=|S|-1$, and let $C$ be a code of length $n$
with parity-check matrix $H$. Finally, let ${\cC}$ be the 
storage code on $G$. Then

(a) If $R({\cC})> 1/2,$ then $n$ is odd and all the rows of $H$ have even weight;

(b) If $R({\cC})> (2^{k}-1)/2^{k},k=2,3,\dots,$ then $(C^\bot)^{\ast (k-1)}\subset C.$
\end{theorem}

\begin{proof}
(a) Lemma~\ref{lem:oddeven} implies that $|S|$ is even and therefore $n$ is odd.
Let $V_i$ be the subspace of $\f_2^r$ generated by all vectors of weight $1$
except the vector with a $1$ in coordinate~$i$. We have $\dim V_i=r-1$. Our assumption of
$R({\cC}) > 1/2$ implies that $\rank A(G)<2^{r-1},$ and then by
Proposition~\ref{prop:odd}, $|S\cap V_i|$ must be even, and so must therefore be
$|S\cap \overline{V_i}|$. But $S\cap\overline{V_i}$ is exactly the set of
columns of $H$ that have a $1$ in row $i$.

(b) We argue by induction. For the base case of $k=2$ we need to show that $C^\bot\subset C.$
Let $i,j$ be distinct (row) indices, $1\leq i,j\leq r.$ Let $V$ be the subspace
of $\f_2^r$ generated by the basis vectors $e_s$ except for $e_i,e_j.$ As above, Proposition~\ref{prop:odd}
implies that $|S\cap \overline{V}|$ is even. In other words, the set of columns of $H$
having a 1 in rows $h_i$ or $h_j$ is even, i.e., $|h_i\cup h_j|$ is even. 
Since the weights of the rows $h_i$ and $h_j$ are also even by Part (a), this implies
that $|h_i\cap h_j|$ is even, i.e., $\langle h_i,h_j\rangle=0.$ The row space of $H$ is
therefore included in its orthogonal space, proving the claim.

We will prove the induction step for $k=3$ to ease notation.
Take three distinct row indices $i_1,i_2,i_3$ and let $V$ be generated by all the basis
vectors except $e_{i_j}, j=1,2,3.$ As before, $|S\cap \overline V|$ is even,
or in other words, $|\cup_{j=1}^3 h_{i_j}|$ is even. Now $|h_{i_1}\cap h_{i_2}|$ and
$|h_{i_1}\cap h_{i_3}|$ are even by the base case, and therefore $|h_{i_1}\cap h_{i_2}\cap h_{i_3}|$
is also even, i.e., $(h_{i_1}\ast h_{i_2}, h_{i_3})=0.$  This shows that $C^\bot$ is orthogonal
to $C^\bot\ast C^\bot,$ meaning that $C^\bot\ast C^\bot\subset (C^\bot)^\bot=C.$
For general $k$ we just use the fact that the $l$-wise intersections of the rows are even for all $1\le l\le k-1,$
and thus so is the $k$-wise intersection.
\end{proof}

\remove{\begin{corollary}\label{lem:R>3/4}
Under the same hypotheses as Lemma~\ref{cor:R>1/2}, if we suppose furthermore
that the storage code has rate $R>3/4$, then the code $C$ defined by the
parity-check matrix $H$ must contain its dual code $C^\perp$.
\end{corollary}

\begin{proof}
Let $i,j$ be distinct indices, $1\leq i,j\leq r$. Let $V_{ij}$ be the subspace of
$\f_2^r$ generated by the vectors of weight $1$, except for the two vectors
whose supports are $\{i\}$ and $\{j\}$. Proposition~\ref{prop:odd} implies that
that the set of columns of $H$ not in $V_{ij}$ is of even cardinality. In other
words, the set of columns of $H$ having a $1$ in rows $i$ or $j$ is of even
cardinality. So the union of the supports of rows $i$ is $j$ is of even
cardinality, and since
rows $i$ are $j$ are of even weight by Corollary~\ref{cor:R>1/2}, this implies
that the two supports intersect in an even number of positions. Equivalently,
this means that row $i$ and row $j$ are orthogonal. .
\end{proof}

\begin{corollary}
Under the same hypotheses as Lemma~\ref{cor:R>1/2}, if we suppose furthermore
that the storage code has rate $R>7/8$, then the code $C$ defined by the
parity-check matrix $H$ must contain $C^\perp*C^\perp$, where, for two codes
$A,B$, $A*B$ denotes the
binary linear code generated by all coordinatewise products $a*b$, $a\in A$,
$b\in B$.
\end{corollary}

\begin{proof}
This time we take three different indices $1\leq i_1,i_2,i_3\leq r$ and define
the subspace $V$ of $\f_2^r$ generated by all vectors of weight $1$ that do not
have supports $\{i_1\},\{i_2\},\{i_3\}$. 
Proposition~\ref{prop:odd} implies that
that the set of columns of $H$ not in $V$ is of even cardinality, which
translates into the union of the supports of rows $i_1,i_2,i_3$ being of even
cardinality. Together with the fact that row $i_3$ must intersect row $i_1$ in
an even number of positions and must also intersect row $i_2$ in an even number
of positions (by Lemma~\ref{lem:R>3/4}), this implies that the intersection of
the three rows must be even. Denoting rows $i_1,i_2,i_3$ by vectors $x,y,z$,
this translates into the orthogonality of $z$ and $x*y$. So any vector of
$C^\perp$ must be orthogonal to any vector of $C^\perp*C^\perp$, meaning
$C^\perp*C^\perp\subset (C^\perp)^\perp=C$.
\end{proof}}

\section{A family of storage codes on triangle-free graphs of rate \texorpdfstring{$3/4$}{}}\label{sec:3/4}
In this section we consider a family of storage codes on coset graphs of a specially
constructed family of binary codes $C_r$ of length $n=2^{r-1}+1$, dimension $k=2^{r-1}-r,$ and distance 4.
To define it, start with the parity-check matrix of the extended Hamming code of
length $2^{r-1}$, augment it with an all-zero column, and then add a row of weight 2
that contains a `1' in the last position. Denote the resulting $(r+1)\times n$ matrix by $H_r.$ For instance, for
$r=4$ we obtain the matrix
 \begin{equation}\label{eq:MH} 
     H_4=\left[\begin{array}{*{8}{c@{\hskip6pt}}|c}
     0&0&0&0&1&1&1&1&0\\
     0&0&1&1&0&0&1&1&0\\
     0&1&0&1&0&1&0&1&0\\
     1&1&1&1&1&1&1&1&0\\
     \cline{1-8}
     0&0&0&0&0&0&0&1&1
     \end{array}\right]
  \end{equation}
where the choice of the last row is largely arbitrary as long as it is of weight 2 and intersects the
all-zero column.  

Observe that the code $C_r$ does not contain its dual $C_r^\bot,$ i.e., the code generated by $H_r,$
because for instance the last row of $H_r$ is not orthogonal to the other rows. Thus, from Theorem~\ref{cor:R>1/2}, the most we can hope
for the storage code constructed on the coset graph of $C_r$ is rate $3/4.$ In this section we prove that 
the rate $R(\cC_r)$ is in fact close to this maximum value.

\begin{theorem} \label{thm:3/4}
Let ${\cC}_r$ be the $[N=2^{r+1},K]$ storage code constructed on the coset
graph of the code $C_r,r\geq 4.$ Then
    $$
    \frac KN=\frac 34-\frac1{2^r}.
    $$
\end{theorem}

The idea behind the construction is that the kernel of the adjacency operator associated to the coset 
graph of the extended Hamming code is easily described, and that the matrix \eqref{eq:MH} is essentially obtained by adding an extra row to the parity-check matrix of the extended Hamming code: we can therefore hope to compute the dimension of the kernel of the adjacency operator associated to this new code from the adjacency operator for the coset graph of the extended Hamming code. Now the latter graph
is simply a complete bipartite graph whose adjacency matrix is
   $$
   A=\begin{bmatrix}0&J\\J&0\end{bmatrix}
   $$
where $J$ is the all-one matrix of order $2^{r-1}.$ \new{This is evident once we realize that the elements of $\f_2^r$ of the form $(\ast,\dots,\ast,0)$ are connected to all the elements $(\ast,\dots,\ast,1)$, while at the same time
having no edges among themselves.} The kernel of $A$ is formed of vectors of length $2^r$
of even weight in the first half as well as in the second half, and therefore has dimension $2^r-2.$ 
To switch from the coset graph of the extended Hamming code to the coset graph of the code $C_r$ we need to understand the effect on the adjacency operators of adding an extra coordinate to a set of generators.

Lemma~\ref{lemma:exact} that we derive below will go some way towards giving a formula for deriving the new dimension of the kernel of the adjacency operator when an extra coordinate is added to the set of generators. To introduce it we need some additional notation.

Let $\tilde{S}$ be any subset of $\f_2^{r+1}$, and let $\tilde{A}$ be the adjacency operator~\eqref{eq:adjacency} on the space of functions $\f_2^{r+1}\to\f_2$, associated to the Cayley graph $\Cay(\f_2^{r+1},\tilde{S})$. 
Let $\pi\colon \f_2^{r+1}\to\f_2^r$ be the projection on the first $r$ coordinates, i.e. $\pi(x)$ is obtained from the vector $x$ by removing its last coordinate.
Now let $S=\pi(\tilde{S})$. Some caution is in order here, because it may happen that two distinct elements of $\tilde{S}$ project onto the same vector if they differ only in their last coordinate. We consider $S$ as a multiset, i.e. we do not collapse into one element the images by $\pi$ of two distinct elements if they happen to have the same projection, therefore we have in particular $|S|=|\tilde{S}|$. Now it may make little sense to consider the Cayley graph over $\f_2^r$ with a generator set $S$ that contains duplicate elements, but we can nevertheless define the operator $A$ on the space of functions $\f_2^r\to\f_2$, associated to $S$ through the formula~\eqref{eq:adjacency}. We note in passing that $A$ is equal to the adjacency operator of the Cayley graph over $\f_2^r$ with the generator set obtained from $S$ by removing duplicate elements.

Next, for $i=0,1,$ let $\tilde{S}_i$ denote the subset consisting of the elements of $\tilde{S}$ with a $i$ in the last coordinate, and let $S_0=\pi(\tilde{S_0})$ and $S_1=\pi(\tilde{S_1})$. Let $A_0$ and $A_1$ be the associated adjacency operators, i.e. the adjacency operators of the Cayley graphs $\Cay(\f_2^r,S_0)$ and $\Cay(\f_2^r,S_1)$.
The operators $A_0$ and $A_1$ act therefore on the space of functions $\f_2^r\to\f_2$.
Note that we have the decomposition:
\[
A = A_0 + A_1.
\]

 Let $e$ denote the vector of weight $1$ in $\f_2^{r+1}$ supported by the last
coordinate,
and let $V$ denote the subspace of $\f_2^{r+1}$ supported by the $r$ old coordinates, 
isomorphic therefore to $\f_2^r$. Recall that $\cF(\f_2^r)$ denotes the space of functions
$\f_2^{\f_2^r}$. Since $\cF(\f_2^{r+1})$ is formed of two copies of $\cF(\f_2^r),$ we have the isomorphism
\begin{eqnarray}
\label{eq:iso}
\F(V)\times\F(V) &\isomto& \F(\f_2^{r+1})\\
  (\one_X,\one_Y)&\mapsto& \one_X+\one_{Y+e}\nonumber
\end{eqnarray}
where $X$ is the support of the first half of the function (vector) and $Y$ the support of its second half.
Now we may identify any function $\tilde{f}\in\F(\f_2^{r+1})$ with a pair of functions $(f_0,f_1)$
through the above isomorphism. Spelling it out:
\begin{align*}
f_0\colon V&\to\f_2\\
          x&\mapsto f_0(x)=\tilde{f}(x)\\
f_1\colon V&\to\f_2\\
          x&\mapsto f_1(x)=\tilde{f}(x+e)
\end{align*}
The matrix $\tilde A$ has the form $\tilde A=\begin{bmatrix}A_0|A_1\\\hline A_1|A_0\end{bmatrix},$ and so we have:
\begin{align}
(\tilde{A}\tilde{f})_0 &= A_0f_0 + A_1f_1\label{eq:0}\\
(\tilde{A}\tilde{f})_1 &= A_0f_1 + A_1f_0\label{eq:1}
\end{align}


We can now derive:

\begin{lemma}\label{lemma:exact}  The adjacency operator $\tilde A$ associated to the set of generators $\tilde S$ satisfies
    \begin{equation}\label{eq:exact}
    \dim\ker\tilde{A} = \dim\ker A +\dim(\ker A_1\cap\ker A) + \dim(\Ima A\cap\Ima
A_{0|\ker A}), 
    \end{equation}
where $ A_{0|\ker A}$ is the restriction of $A_0$ to the kernel of $A$.   
\end{lemma}
\begin{proof}
Let $F_D$ ($D$ for double) be the set of functions in $\F(V)\times\F(V)$ of the form $(f,f)$, and
let $F_L$ ($L$ for left) be the set of functions of the form $(f,0)$ and note that $F_D\cap F_L=\{0\}.$
 Any function $\tilde{f}\in\F(\f_2^{r+1})$, identified with $\F(V)\times\F(V)$ through the
isomorphism~\eqref{eq:iso}, has a unique decomposition as
$\tilde{f}=\tilde{f}_D+\tilde{f}_L$, with $\tilde{f}_D\in F_D$ and
$\tilde{f}_L\in F_L$. Therefore, $\tilde{f}\in\ker\tilde{A}$ if and only if 
$\tilde{A}\tilde{f}_D=\tilde{A}\tilde{f}_L$. We therefore have
\[\dim\ker\tilde{A}=\dim\ker\tilde{A}_{|F_D}+\dim\ker\tilde{A}_{|F_L}+\dim(\Ima\tilde{A}_{|F_D}\cap\Ima\tilde{A}_{|F_L}).\]
On account of \eqref{eq:0} and \eqref{eq:1} any function $\tilde{f}\in F_D$ is in
$\ker\tilde{A}$ iff $(A_0+A_1)f_0=Af_0=0$, i.e. $f_0\in\ker A$. Therefore
  $$
  \dim\ker\tilde{A}_{F_D}=\dim\ker A.
   $$
 Next, we
consider functions $\tilde{f}$ for which $f_1=0$. 
Again by \eqref{eq:0} and \eqref{eq:1} a function $\tilde f\in F_L$ is in
$\ker\tilde{A}$ iff $f_0\in(\ker A_0\cap\ker A_1)$. From $A=A_0+A_1$ we have that
$\ker A_0\cap\ker A_1=\ker A_1\cap \ker A,$ and so
\[
\dim\ker\tilde{A}_{F_L}=\dim(\ker A_1\cap\ker A).
\]
It remains to compute the dimension of
$\Ima\tilde{A}_{|F_D}\cap\Ima\tilde{A}_{|F_L}$.
Let $\tilde{f}= (f,f)\in F_D,$ then from \eqref{eq:0} and \eqref{eq:1} we have 
$\tilde{A}\tilde{f}\cong (Af,Af).$

Now let $\tilde{f'}= (f',0)\in F_L$, then, again from  \eqref{eq:0}-\eqref{eq:1},
$\tilde{A}\tilde{f'}\cong (A_0f',A_1f')$. 
So $\tilde{A}\tilde{f'}\in\Ima F_D$ if
and only if $A_0f'\in\Ima A$ and $A_0f'=A_1f'.$ This last condition is
equivalent, since $A_0+A_1=A$, to $f'\in\ker A$.
Therefore, $\Ima F_D\cap\Ima F_L \simeq \Ima A\cap\Ima A_{0|\ker A}$.
\end{proof}

\noindent{\em Proof of Theorem \ref{thm:3/4}:}
Let $\tilde{S}\subset\f_2^{r+1}$ consist of the columns of the parity-check matrix $H_r$ of $C_r$ to which we add the zero vector, so that the storage code associated to $C_r$ is the kernel of the operator $\tilde{A}$ associated to $\tilde{S}$.
We now apply Lemma~\ref{lemma:exact} to $\tilde{A}$ and
compute the dimensions of the spaces on the right-hand side of \eqref{eq:exact}.

Removing the last coordinate to the elements of $\tilde{S}$ yields a set $S$ equal to the set of columns of the parity-check matrix of the extended Hamming code, plus two copies of the zero vector. The associated operator $A$ is therefore equal to the adjacency operator of the Cayley graph over $\f_2^r$, with generator set equal to all vectors of $\f_2^r$ that end with a $1$ on their last coordinate.
As noted earlier, this graph is the complete bipartite graph on $\f_2^{r-1}$, and
the kernel of $A$ is the space of functions that are of even weight on each of the spaces $V_0$ and $V_1$, 
corresponding to the set of vectors $(x_1,\ldots,x_r)$ satisfying $x_r=0$ and $x_r=1$, and its dimension equals $$\dim \ker A=2^r-2.$$
The last row of the matrix $H_r$ is of weight 2, 
so there are exactly two elements of $\tilde{S}$ that end with a 1 in their last coordinate, which means that $|S_1|=2.$ Furthermore, one of these two elements of $\tilde{S}$ projects onto the zero vector of $\f_2^r$ by $\pi$, so that $S_1=\{0,s\}$ for some non-zero vector $s\in\f_2^r,$ e.g., in \eqref{eq:MH} $s=\one^r$. Applying \eqref{eq:adjacency} we have
\[
A_1f(x)=f(x)+f(x+s).
\]
We therefore have that $\ker A_1$ is equal to the space of functions whose support in $\f_2^r$ is stable by addition by $s$. Now, since $s$ by construction must be equal to a column of the parity-check matrix of the extended Hamming code, it ends with a 1, i.e. is of the form $s=(s_1,\ldots,s_{r-1},s_r=1)$. Therefore, a function is in $\ker A_1$ if and only if its
support is equal to $X\cup (X+s)$ for $X\subset V_0$ and $X+s\subset V_1.$
Intersecting $\ker A_1$ with $\ker A$ restricts the sets $X$ to sets of even size, and thus
$$\dim(\ker A_1\cap\ker A)=\dim\ker A_1 -1=2^{r-1}-1.$$

For the third term in \eqref{eq:exact} observe that $A_1\one_{V_0}=\one_{V_0}+\one_{V_1}=\one_V\in\Ima A\cap\Ima A_{1|\ker A}.$ Furthermore, we clearly have that $\Ima A_1$
must consist of functions whose support is stable by addition by $s$, i.e. $\Ima A_1=\ker A_1$, implying that $\one_{V_0}$
is in $\Ima A$ but not in $\Ima A_1.$ Since $\dim(\Ima A)=2$,
we have therefore that $\Ima A\cap\Ima A_{1|\ker A}$ is
equal to the space of constant functions and of dimension $1$.
Since $A_0+A_1=A$, we have that $A_0$ and $A_1$ are equal on $\ker A$, and
$$\dim(\Ima A\cap\Ima A_{0|\ker A})=1.$$

Collecting the terms in \eqref{eq:exact}, we obtain
   $$
   \dim\ker \tilde A=2^r+2^{r-1}-2,
   $$
which yields the rate expression in the theorem.
\qed

\section{Correcting multiple erasures}\label{sec:multiple}

\subsection{The edge-vertex construction and graph expansion} 
The most likely failure scenario in storage applications is loss of a single node, which corresponds to
correcting a single erasure in the codeword of a storage code. This question is central to the paper that
defined codes with local erasure correction \cite{gopalan2011locality}, and it was also one of the original motivations for studying such codes on graphs in the paper by Mazumdar \cite{Mazumdar2015}. At the same time, correcting multiple erasures represents
a natural extension of this problem, which was addressed in a large number of papers on codes with locality, among them \cite{Kamath2014,Tamo2013}. Motivated by this research, \cite{Mazumdar2015} posed the question of recovering the codeword when multiple vertices are erased. This question can be posed in several ways based
on the vertex recovery procedure, and we proceed to specify our assumptions. Throughout this section
we limit ourselves to $d$-regular graphs and to codes obtained from the edge-vertex construction of Section~\ref{sec:constructions}.

We begin with a code $\cC(G)=\{x, x\in \f_q^{d\cdot|V|}\}$ defined on a graph $G(V,E)$, where as before, the coordinates of the codeword are placed on the vertices of $G$. The coordinates $x_v$ can be viewed either as $d$-vectors over $\f_q$ or as
elements of $\f_{q^d}$, and they are obtained by collecting all the values placed on the edges incident to $v.$ Extending the definition of storage codes, suppose that every vertex $v\in V$ can recover its value $x_v$ from the values of its neighbors $x_u, u\in \cN(v)$ as long as at most $t$ of them are unavailable. Of course, taking $t=0$ takes us back to the original definition of storage codes.
The easiest form of the multiple erasure correction problem arises if we assume that every erased vertex is adjacent to at most $t$ other erased vertices. To address it, we simply place constraints on the edges in the neighborhood
of every vertex: namely, assume that the symbols placed on the edges in $E(v)$ form a vector in a linear code ${D}$ of length $d$ over $\f_q$ that corrects $t$ erasures. This defines a linear storage code $\cC=\cC(G,{D})$ that enables each vertex to recover its value from $d-t$ or more nonerased neighbors. 
The dimension of the code $\cC$ is at least $|V|(2R({D})-1),$ where $R({D}):=\frac{\dim({D})}{d}$ is the rate of the local code ${D}.$

This approach allows correction only of erasure patterns that leave $d-t$ or more neighbors of each
vertex nonerased, which is a somewhat artificial assumption \new{(indeed, the erased vertices may not follow the connections in the graph)}. Lifting it, we will assume next that the set
of erased vertices $S\subset V$ is of arbitrary shape, and engage standard ideas from error-correcting codes on graphs, in particular,
graph expansion and its use in decoding \cite{SipserSpielman1996}. Namely, instead of specifying that the vertices
can correct themselves independently, we will be satisfied if there is one erased vertex with $t$ or fewer neighbors in the subgraph induced by $S$. 
Once it is recovered, there are fewer erased vertices (edges), and (under some conditions)
there will be more erased vertices that can correct themselves from their neighborhoods.

As remarked, this approach is close to the classic error-correcting codes on graphs, also known as the Tanner codes
\cite{Tanner1981}. The only difference between them and our construction of storage codes is related to the way erasures
are applied to the codeword coordinates. Namely, while in earlier works the edges were erased independently of each other based on the properties of the communication channel, in storage codes erasures affect the vertices, which results in erasing all the values of the edges incident to the vertex at once.

The edge-vertex construction enables us to iteratively correct multiple erasures if we assume that the underlying graph
has expansion properties.
We will use the following well-known result.

\begin{lemma}\label{lemma:EML}{\rm (Expander mixing lemma \cite[Lemma 2.5]{Hoory2006})} Let $G$ be a $d$-regular graph with $N$ vertices and eigenvalues $\lambda_1=d\ge \lambda_2\ge \dots\ge \lambda_N.$ Then for any
$U,T\subset V$
   $$
   \Big||E(U,T)|-\frac{d |U||T|}{N}\Big|\le \lambda\sqrt{|U||T|},
   $$
   where $\lambda:=\max(|\lambda_2|,|\lambda_N|)$.
\end{lemma}

\begin{proposition} Consider a storage code $\cC(G,{D})$ defined by the edge-vertex construction. 
Suppose that $G$ is a $d$-regular graph with $N$ vertices and the local code ${D}$ corrects $t$ erasures. Let $U\subset V, |U|=\sigma N$ be the set of erased vertices. As long as 
   \begin{equation}\label{eq:S}
   \sigma\le \frac td-\frac{\lambda}d,
   \end{equation}
the erased vertices can be recovered based on the local iterative procedure.
\end{proposition}
\begin{proof} Taking $T=U$ in Lemma \ref{lemma:EML}, we obtain for the number of edges in the subgraph induced by $U$
the inequality
   $$
   2|E(U)|\le \frac dN |U|^2+\lambda |U|.
   $$
Let $\partial (U)=\{(u,v)\in E(G): u\in U, v\in V\backslash U\}$ be the {\em edge boundary}
of $U$. We have
   $$
   d |U|=2 |E(U)|+|\partial(U)|.
   $$
Taken together, this implies
  $$
  \frac{|\partial(U)|}{\sigma N}\ge d(1-\sigma)-\lambda.
  $$
In other words, there exists a vertex $v\in U$ with at least $d(1-\sigma)-\lambda$ nonerased neighbors. As long
as 
   $$
   d(1-\sigma)-\lambda\ge d-t,
   $$
its value can be recovered using the local code ${D}.$ Under the assumption \eqref{eq:S} this inequality is
satisfied, so size of the erased set of vertices has decreased. The remaining set of erased
vertices $U'=U\backslash \{v\}$ also satisfies \eqref{eq:S}, and the proof
is concluded by straightforward induction.
\end{proof}

To use this result, we need the spectral gap $d-\lambda$ to be large compared to $d$ or $\lambda$ much smaller
than $d.$
The limits for the spectral gap are given by the Alon-Boppana bound \cite[Thm.2.7]{Hoory2006}, namely for any $d$-regular graph on $n$ vertices
   $$
   \lambda\ge 2\sqrt{d-1}-o_n(1).
   $$
The graphs with $\lambda\le 2\sqrt{d-1}$ are known to exist by the Lubotzky-Phillips-Sarnak and Margulis constructions
\cite[Sec.5.11]{Hoory2006}. Assuming that the graph $G$ is from this family, we can state the following:
\begin{corollary} There exist storage codes $\cC(G,{D})$ on $d$-regular graphs with local codes correcting $t$ erasures
that recover from any pattern of $s$ erased vertices as long as their proportion $\sigma=s/n$ satisfies
   $$
   \sigma \le \frac td-O(d^{-1/2}).
   $$
\end{corollary}
Note that large spectral gap guarantees that a subset $U$ has many edges that connect it with its complement, or in other words, its {\em edge expansion ratio} is large. Further connections between the spectral gap and expansion are discussed in \cite[Sec.~4.5]{Hoory2006}.

\section{Discussion}\label{sec:discussion}
\subsection{Numerical experiments}
It is tempting to look for other coset graphs of linear codes whose full-parity storage code has rate greater than $1/2$. It is a challenge to derive general formulae for dimensions however. With computer help we find:

\begin{proposition}\label{prop:numerical} \hspace{1cm}

(a) The binary Golay code of length $23$ and dimension $11$ yields a graph $G$ on $N=2048$ vertices that supports a linear storage code of rate $41/64$.

(b) The 2-error-correcting binary BCH code of length $n=2^s-1$ and dimension $k=2^s-1-2s$ yields a graph $G_s$
with $N=2^{2s}$ vertices that supports a linear storage code with rate given in the following table
\begin{center}\begin{tabular}{c@{\quad}cccccc}
$s$&$4$&$5$&$6$&$7$&$8$\\
$R_2(G_s)$&$\frac{39}{64}$&$\frac{347}{512}$&$\frac{1497}{2048}$&$\frac{6387}{8192}$&$\frac{26859}{32768}$.\\[.1in]
\end{tabular}\end{center}
\end{proposition}

The sequence of rate values obtained in Proposition~\ref{prop:numerical} is {\small $0.6094, 0.6777, 0.7309, 0.7796, 0.8196,$} and the value
$R_2(G_8)=0.8196$ of the storage code of length $N=65536$ now represents the largest known rate of storage codes on triangle-free graphs over any alphabet.

\subsection{Open problems} \hspace{1cm}

1. {\em The rate question:}
Arguably, the central question is whether rates of storage codes on triangle-free graphs can be arbitrarily close to 1. If not, what is an upper limit for these rates ?

2. {\em BCH codes:} In view of the numerical experiments it is of interest to find or estimate the dimension of storage codes obtained from the 
2-error-correcting BCH codes. More specifically, what is $\limsup_{s\to\infty}R_2(G_s)$? At this point we cannot even rule out
that the sequence $(R_2(G_s))_s$ in the limit reaches 1, which would resolve the question of the maximum possible rate of storage codes on triangle-free graphs.

3. {\em Reed-Muller codes:} \new{Another good candidate arises from the family of Reed-Muller codes. Fix $m\ge 4$
and consider Boolean polynomials $f(v):\f_2^m\to \f_2.$ The second order Reed-Muller code $RM(m,2)$ is spanned by the
set of evaluations of functions of degree $\le 2$ on all the elements of $\f_2^m.$ Consider a subcode $C_m$ of the 
punctured code $RM(m,2)$ spanned by the evaluations of the functions $v_i, 1\le i\le m$ and $v_iv_j, 1\le i<j\le m$ on the nonzero points of $\f_2^m$. The length of the code $C_m$ is $2^m-1$ and its dimension is $m(m+1)/2$. Another way of constructing this code is to take a linear space spanned by the Simplex code $S_m$ and its Schur square $S_m\ast S_m.$ 

Now consider the code $(C_m)^\bot$, i.e., a code whose parity check matrix $H$ has rows given by the evaluations of the linear and quadratic functions. This is a code of odd length with even-weight parities, and it also contains its 
dual as well as Schur powers of the dual, fulfilling the necessary conditions for high-rate storage codes in 
Theorem \ref{cor:R>1/2}. Finding the actual rate of the storage code $\cC$ is however not straightforward, and we leave this as an open problem.}

\subsection{Erasure correction and bootstrap percolation}
We conclude with one more observation regarding correcting multiple erasures with storage codes. Given a graph $G$, we again assume the edge-vertex construction that relies on a local code correcting $t$ erasures.
Assume that the vertices are erased randomly and independently with some probability $\bar p,$ 
forming a set $U\subset V$ of erased  vertices. We wish for the iterative procedure employed in the previous section to successively recover the erased vertices until all are corrected. Rephrasing, suppose that functional vertices are selected randomly with probability $p=1-\bar p,$ and an erased vertex with not more than $t$ erased neighbors recovers its value, becoming functional. In graph-theoretic terms this procedure is known as 
{\em bootstrap percolation,} and it represents an established branch of percolation theory. 
Percolation occurs when all the erased vertices have been recovered (for infinite graphs, recovered with probability one). This problem was introduced in \cite{Chalupa1979}; see the recent paper \cite{Gravner2015} for an
overview of the main results. The main problem studied in the literature is the determination of the
{\em critical probability}, defined as 
    $$
    p_c:=\inf\{p: {\mathbb P}_p(U \text{ is corrected by the iterative process})\ge 1/2\}.
    $$
The known results include the determination of $p_c$ for the infinite grid $[n]^d$ as $n\to\infty$ (see \cite{Holroyd2003} for $d=2$ and \cite{Balogh2012} for all $d$) as well as for several other classes of graphs.
These results translate directly to the corresponding thresholds for erasure correction with storage codes
on graphs in which the vertices are erased randomly and independently.    
    
Bootstrap percolation has been also considered in the deterministic setting \cite{CojaOghlan2015,Freund2018}, 
where the main question is finding the minimum size of a set of functional vertices that enables the code to 
correct all erasures, or, using the language of the cited papers, the maximum number of {\em infected vertices} 
that infect the entire graph. Rephrasing, this means that we attempt to pinpoint a particular combination of erasures of the largest possible size that can be recovered through the iterative $t$-neighbor process. 
This problem, however, is an opposite of the natural question addressed by storage codes, where one is interested in the largest $s$ such that {\em any} combination of $s$ erased vertices can be recovered, or the smallest size of the set of nonerased vertices that enable recovery of the entire graph irrespective of the shape of the erased set $U.$ 

\subsection*{\sc Acknowledgment} The research of Alexander Barg was partially supported by NSF grants CCF2110113 and CCF2104489.

\end{document}